\documentclass[letterpaper, 10pt, conference]{ieeeconf}

\usepackage{cite}
\usepackage{amsmath,amssymb,amsfonts}
\usepackage{graphicx}
\usepackage{textcomp}
\usepackage{stmaryrd}
\usepackage{tabularx}
\usepackage{multirow}
\usepackage{float}
\usepackage{mathtools}
\usepackage{comment}

\usepackage{color}

\newcommand{\tnf}[1]{\textnormal{#1}}
\newcommand{\tbf}[1]{\textbf{#1}}

\newcommand{\R}{\mathbb{R}}
\newcommand{\N}{\mathbb{N}}

\newcommand{\norm}[1]{\left\lVert{#1}\right\rVert}

\newcommand{\bmat}[1]{\begin{bmatrix}#1\end{bmatrix}}

\newtheorem{thm}{Theorem}
\newtheorem{defn}[thm]{Definition}
\newtheorem{lem}[thm]{Lemma}

\newtheorem{cor}[thm]{Corollary}

\floatstyle{ruled}
\newfloat{block}{thbp}{lop}
\floatname{block}{Block}

\floatstyle{ruled}
\newfloat{block}{thbp}{lop}
\floatname{block}{Block}

\usepackage{upgreek}

\newcommand\blfootnote[1]{%
	\begingroup
	\renewcommand\thefootnote{}\footnote{#1}%
	\addtocounter{footnote}{-1}%
	\endgroup
}

\title{\LARGE \bf
	An Exact Lyapunov Characterization of Regional Rate Performance\\ for Rational Stability
}

\author{Declan S. Jagt, Matthew M. Peet %
}

\begin{document}

	\maketitle
	\pagestyle{plain}

\begin{abstract}
	Rational stability is a quantitative stability notion for nonlinear ordinary differential equations (ODEs), under which convergence is characterized by a rationally-in-time decaying bound on solutions. Although Lyapunov characterizations of rational stability have been proposed, no existing condition has been shown to exactly characterize rational rate performance.  The paper resolves this issue by providing an exact converse Lyapunov characterization of regional rational stability with prescribed rate performance. To test the resulting converse Lyapunov conditions numerically, a tiered set of convex relaxations is proposed which can be enforced using Sum-of-Squares (SOS) programming. A conjectured conservatism scaling factor in rate performance is associated with each tier and validated numerically. The first proposed tier appears to be novel and is shown to outperform classical SOS-based approaches. Numerical examples are used to validate the results and compute nested regions of state-space on which prescribed levels of performance can be guaranteed.
\end{abstract}

\blfootnote{\vspace*{-0.00cm}%
	\tbf{Acknowledgment:} This work was supported by National Science Foundation grants 2429973 and 2606245.} 


\vspace*{-0.2cm}
\section{Introduction}

Nonlinear ordinary differential equations (ODEs) are used to model a wide range of physical and engineered systems. Stability analysis of such systems typically seeks to establish convergence of solutions to an equilibrium. In many applications, however, it is also important to quantify the \emph{rate} of convergence and, for systems that are only locally stable, the region over which a prescribed rate can be guaranteed. Quantitative notions of stability address these questions by imposing pointwise-in-time bounds on solutions.

Perhaps the best-known quantitative stability notion is exponential stability, which requires solutions to be bounded by an exponentially decaying function. However, many asymptotically stable nonlinear systems are not exponentially stable; in particular, local exponential stability requires the linearization at the equilibrium to be Hurwitz. For systems exhibiting slower convergence, a more appropriate notion may be \emph{rational stability}\footnote{also referred to as weakly intensive behavior~\cite{hahn1963LFs}, and closely related to the notion of polynomial stability as in~\cite{caraballo2001decay}.}, under which solutions satisfy
\begin{equation}\label{eq:rational_bound_intro}
	\norm{x(t)}_2
	\leq
	M^{1/p}\frac{\norm{x(0)}_2^\eta}{\sqrt[p]{1+kt\norm{x(0)}_2^p}},
\end{equation}
for constants $M,k,p,\eta>0$. For $x(0)\neq0$, this bound implies the asymptotic decay $\norm{x(t)}_2=\mathcal{O}(t^{-1/p})$, while additionally quantifying the dependence of the transient behavior on the initial condition. The parameters in~\eqref{eq:rational_bound_intro} can therefore be used to quantify the performance of rationally stable systems. For example, with $M$, $p$, and $\eta$ fixed, the largest admissible value of $k$ can be interpreted as the \emph{rate performance} of the ODE.


Several Lyapunov characterizations have been proposed for verifying rational stability. A classical sufficient condition is the existence of a function $V$ satisfying
\begin{equation}\label{eq:LF_conservative_intro}
	C_{1}\norm{x}_{2}^{r_{1}}\leq V(x)\leq C_{2}\norm{x}_{2}^{r_{2}},
	\quad
	\dot{V}(x)\leq -C_{3}\norm{x}_{2}^{r_{3}},
\end{equation}
for suitable constants $C_i$ and $r_i$~\cite{bacciotti2005LF_book}. Alternatively, rational stability can be certified by a function $V$ satisfying
\begin{equation}\label{eq:LF_alpha_intro}
	c_{1}\norm{x}_{2}^{s_{1}}\leq V(x)\leq c_{2}\norm{x}_{2}^{s_{2}},
	\quad
	\dot{V}(x)\leq -c_{3}V(x)^{1+s_{3}},
\end{equation}
for suitable constants $c_i$ and $s_i$~\cite{jammazi2013rational}. Both sets of conditions provide equivalent Lyapunov characterizations of rational stability: feasibility of either set for suitable constants implies a solution bound of the form~\eqref{eq:rational_bound_intro} for some $p$, $M$, $k$, and $\eta$, while rational stability conversely implies the existence of Lyapunov functions satisfying each set of conditions~\cite{bacciotti2005LF_book,jammazi2013rational}. Thus, these characterizations determine whether an ODE is rationally stable for \emph{some} choice of the parameters in~\eqref{eq:rational_bound_intro}.

In the context of \emph{quantifying} stability performance, however, two additional questions arise. First, are these Lyapunov conditions also necessary for a \emph{prescribed} level of performance on a prescribed region? That is, if~\eqref{eq:rational_bound_intro} holds for specified $M$, $k$, $p$, and $\eta$ and for initial conditions in a region $G$, does there necessarily exist a Lyapunov function certifying the same performance on the same region? Second, how can such a Lyapunov function be found computationally with minimal conservatism? The central objective is to address these two questions by developing necessary and sufficient Lyapunov conditions for certifying prescribed rational stability performance, together with a computational test for their verification.

Regarding the first question, although converse results have been established for the conditions in~\eqref{eq:LF_conservative_intro} and~\eqref{eq:LF_alpha_intro}, these results do not recover the exact performance parameters in~\eqref{eq:rational_bound_intro}. For example, the proof of~\cite[Thm.~5.3]{bacciotti2005LF_book} shows that feasibility of~\eqref{eq:LF_conservative_intro} implies rational stability with rate $k=(\frac{r_3}{r_2}-1)\frac{C_3}{C_2}$, whereas, starting from a rational bound with rate $k$, the converse construction yields constants satisfying $(\frac{r_3}{r_2}-1)\frac{C_3}{C_2}=\frac{k}{M^2}$. Thus, the converse construction need not recover the original rate when $M>1$. Furthermore, necessity of the Lyapunov characterization in~\eqref{eq:LF_alpha_intro} has been established through its relationship with the characterization in~\eqref{eq:LF_conservative_intro}, rather than through an explicit converse construction that preserves the prescribed performance parameters.

To address this question, we adapt an approach developed by Gr\"une for quantitative input-to-state stability analysis~\cite{grune2002KLD_functions,grune2004KLD,grune2004quantitativeISS}. In this framework, trajectory bounds are expressed in terms of a state measure $\alpha$ and a comparison function $\beta$ as $\alpha(x(t))\leq\beta(\alpha(x(0)),t)$, where $\beta$ is the solution map of a scalar differential equation. The comparison principle can then be used to construct converse Lyapunov functions that certify the prescribed trajectory bound up to arbitrary accuracy. To apply this framework to rational stability, we fix $\eta=1$ and define $\alpha(x):=\norm{x}_2^p$ and $\beta(y,t):=\frac{y}{1+yt}$, where $\beta$ is the solution map of $\dot y=-y^2$. The rational stability bound in~\eqref{eq:rational_bound_intro} can then be equivalently expressed as
\begin{equation*}
	\alpha(x(t))\leq M\beta\left(\alpha(x(0)),kt\right).
\end{equation*}
Adapting the construction in~\cite{grune2002KLD_functions} yields a necessary and sufficient Lyapunov characterization for fixed $p$, $k$, and $M$, similar in form to~\eqref{eq:LF_alpha_intro}. To ensure continuity of the converse Lyapunov function, we allow an arbitrarily small relaxation of the rate and establish the existence of a suitable converse function for every $\bar{k}\in[0,k)$. We further show that this function need only be defined on an appropriate invariant region $G$, yielding a regional converse characterization tied directly to the prescribed values of $p$, $k$, and $M$. Such a Lyapunov characterization is also developed for more general quantitative stability notions in the preprint~\cite{jagt2026ConverseLF}---here, we focus on and develop in detail the case of rational stability.


The second question concerns the computational verification of such performance guarantees. The classical conditions in~\eqref{eq:LF_conservative_intro} have an important advantage in this regard: they are affine in the unknown function $V$ and are therefore amenable to convex optimization. Moreover, for polynomial ODEs, existence of a Lyapunov function satisfying these conditions implies existence of a polynomial Lyapunov function~\cite{leth2017rational_stability}, making them well suited to Sum-of-Squares (SOS) programming. However, these conditions may be conservative when certifying a prescribed level of rate performance. In contrast, the performance-oriented conditions in~\eqref{eq:LF_alpha_intro} more closely capture the rational decay bound but are nonlinear in $V$, making computational search considerably more difficult. Verification of these conditions has instead typically relied on problem-specific Lyapunov constructions, with application to several classes of nonlinear control systems, including bilinear systems~\cite{zaghdoudi2016rational}, systems with drift~\cite{zaghdoudi2017partial_rational_stabilizability}, time-varying systems~\cite{jammazi2019rational_stabilization}, and time-delay systems~\cite{echi2018rational_timedelay}.\footnote{Related Lyapunov arguments have also been used to characterize polynomial stability and decay bounds for homogeneous systems in~\cite{nakamura2002smoothLFs,sun2022polynomial_stability,jammazi2021rational_observation}.} This reliance on problem-specific constructions limits the direct use of these conditions for computational stability analysis and controller synthesis for more general nonlinear ODEs.

We therefore seek Lyapunov conditions that retain the computational tractability of~\eqref{eq:LF_conservative_intro} while reducing its conservatism in rate performance. To this end, we introduce the intermediate dissipation condition
\begin{equation*}
	\dot{V}(x)\leq -kV(x)\norm{x}_{2}^{p}.
\end{equation*}
This condition remains affine in the unknown Lyapunov function \(V\) and can be tested using convex optimization, while more directly reflecting the state-dependent decay in~\eqref{eq:rational_bound_intro}. We compare its conservatism with that of the classical and performance-oriented Lyapunov conditions, introduce a quantitative measure of this conservatism, and develop an SOS formulation for computational verification. An overview of the methodology is provided in Subsection~\ref{subsec:rate_performance}.



\section{Problem Formulation and Methodology}

\subsection{Notation}

Let $\R_{+}:=[0,\infty)$. Denote by $\R_{d}^{m}[x]$ the set of $\R^m$-valued polynomials of degree at most $d\in\N$ in $x\in\R^{n}$.
For a given vector field $f:\Omega\to\R^{n}$ on a domain $\Omega\subseteq\R^{n}$, let $\phi_{f}(x,t)$ denote the solution at time $t$ of the corresponding ODE initialized at $x$, so that
\begin{align*}
	\frac{\partial}{\partial t}\phi_{f}(x,t)
	&=f(\phi_{f}(x,t)),
	&
	\phi_{f}(x,0)&=x.
\end{align*}
Throughout the paper, we assume that the solutions under consideration exist uniquely for all $t\in\R_{+}$ and depend continuously on their initial conditions. Standard conditions ensuring local existence, uniqueness, and continuous dependence are given, e.g., in~\cite[Thms.~3.1,~3.4]{khalil2002nonlinear}. We say that $G\subseteq\Omega$ is forward invariant if $x\in G$ implies $\phi_{f}(x,t)\in G$ for all $t\geq 0$.
Uniqueness of solutions ensures that $\phi_f$ satisfies the semigroup property
\begin{equation*}
	\phi_f(\phi_f(x,t),s)
	=
	\phi_f(x,t+s),
	\qquad
	\forall\,t,s\in\R_+.
\end{equation*}
For a continuous function $v:\R\to\R$, we denote by $D_{t}^{+}v(t)$ the upper right-hand (Dini) derivative of $v(t)$ at $t$.
For a continuous function $V:\Omega\to\R$, we then define $\dot{V}(x):=D_{t}^{+}V(\phi_{f}(x,t))|_{t=0}$, so that
\begin{equation*}
	\dot V(x)
	=
	\limsup_{\Delta t\to0^+}
	\frac{V(\phi_f(x,\Delta t))-V(x)}{\Delta t}.
\end{equation*}

\subsection{Problem Formulation: Quantifying Rate Performance for Rationally Stable Systems}\label{subsec:rate_performance}

The goal is to establish necessary and sufficient Lyapunov conditions for quantifying the rate of decay of a rationally stable vector field $f$ on a prescribed invariant set $G$, and to develop Sum-of-Squares (SOS) conditions for computationally verifying this rate performance. To this end, we first formalize the notion of rational stability and its associated rate performance.

Traditionally, a vector field $f$ is said to be rationally stable if there exist $M\geq 1$, $p,\delta>0$, and $\eta\in(0,1]$ such that $\norm{x}_2<\delta$ implies~\cite{bacciotti2005LF_book}
\begin{equation}\label{eq:rational_bound_old}
	\norm{\phi_f(x,t)}_2
	\leq
	M\frac{\norm{x}_2^\eta}
	{\sqrt[p]{1+\norm{x}_2^p t}},
	\qquad \forall t\geq0.
\end{equation}
While this definition establishes rational stability, it requires only the existence of \emph{some} region $G:=\{x\in\R^n\mid\norm{x}_2<\delta\}$ and \emph{some} parameters $M$, $p$, and $\eta$ for which~\eqref{eq:rational_bound_old} holds. It is thus not directly suited to quantifying a prescribed level of performance on a given region. Therefore, we instead consider the following modified notion of rational stability.

\begin{defn}\label{defn:rational_stability}
	For given $G\subseteq\R^n$, $p,k>0$, and $M\geq1$, we say that the vector field $f:\R^n\to\R^n$ is \tbf{rationally stable} on $G$ with power $p$, rate $k$, and gain $M$ if
	\begin{equation*}\label{eq:rational_bound}
		\norm{\phi_f(x,t)}_2^p
		\leq
		M\frac{\norm{x}_2^p}
		{1+\norm{x}_2^p kt},
		\qquad \forall x\in G,~t\geq0.
	\end{equation*}
\end{defn}
\vspace*{2mm}

This definition differs from the conventional definition in~\eqref{eq:rational_bound_old} in three respects. First, rather than requiring the existence of a region on which the rational bound holds, Defn.~\ref{defn:rational_stability} considers a prescribed region $G$. Second, rather than requiring the existence of parameters for which the bound holds, it considers prescribed values of the power, rate, and gain. In doing so, we fix $\eta=1$, since the bound in~\eqref{eq:rational_bound_old} cannot hold globally for $\eta\in(0,1)$, as is evident by considering $t=0$. Finally, Defn.~\ref{defn:rational_stability} introduces an explicit rate parameter $k$ that scales the temporal decay of the bound. For fixed $M$ and $p$, increasing $k$ yields a uniformly tighter bound for every $x\neq0$ and $t>0$, and therefore corresponds to faster guaranteed convergence. This motivates the following notion of rational rate performance.

\begin{defn}
	For given $G\subseteq\R^n$, $p>0$, and $M\geq1$, we define the rational stability \tbf{rate performance} of $f:\R^n\to\R^n$ on $G$ as the supremum of all $k\geq0$ such that $f$ is rationally stable on $G$ with power $p$, rate $k$, and gain $M$.
\end{defn}

Although we use the rate parameter $k$ as the primary measure of performance, the power parameter $p$ in Defn.~\ref{defn:rational_stability} may also be used to quantify the convergence of rationally stable systems.\footnote{The parameter $p$ may also be referred to as the rate of decay~\cite{jammazi2013rational,zaghdoudi2017partial_rational_stabilizability}.} Unlike $k$, however, $p$ does not provide a consistent ordering of convergence rates over arbitrary domains $G$. In particular, for $\norm{x}_2\geq1$, increasing $p$ implies faster decay for $t<1$ but slower decay for $t>1$. In contrast, for fixed $M$ and $p$, increasing $k$ monotonically tightens the rational bound for every $x\neq0$ and $t>0$. Thus, $k$ provides a natural measure of rational rate performance and will be our primary focus, although the results can also be used to assess performance with respect to $p$.

With this notion of rate performance established, Section~\ref{sec:LF_characterization} develops Lyapunov conditions, similar in form to~\eqref{eq:LF_alpha_intro}, that are sufficient to certify prescribed rational performance on an invariant set. Using a Yoshizawa-type construction similar to that in~\cite{grune2002KLD_functions}, we establish a converse result showing that these conditions are also necessary, recovering the rational rate performance up to arbitrary accuracy on the same invariant set. Section~\ref{sec:SOS_conditions} then develops a convex relaxation of this performance-exact characterization that is affine in the unknown $V$ and can therefore be enforced using SOS programming for polynomial systems. We characterize the conservatism of this relaxation relative to both the performance-exact Lyapunov conditions and the classical norm-based conditions in~\eqref{eq:LF_conservative_intro}, and formulate SOS programs for estimating rational rate performance on prescribed regions. Finally, Section~\ref{sec:examples} applies these programs to several rationally stable systems to numerically examine the predicted conservatism and compute regions on which different levels of rational performance can be guaranteed.

\section{An Exact Lyapunov Function Characterization of Rational Stability}\label{sec:LF_characterization}

Having defined a notion of rate performance for rational stability, we now show how this performance can be characterized through a suitable Lyapunov function. Specifically, we first show that a prescribed rate $k$ on a given invariant region $G$ can be certified by a Lyapunov function $V:G\to\R_{+}$ satisfying $M^{-1}\norm{x}_2^{p}\leq V(x)\leq\norm{x}_2^{p}$ and
\begin{equation}\label{eq:dV_bound}
	\dot{V}(x)\leq -kV(x)^2,\qquad \forall x\in G.
\end{equation}
We then establish a converse result showing that rational stability with rate $k$ on $G$ implies the existence of a continuous function satisfying the corresponding Lyapunov conditions up to arbitrary accuracy in the rate. In the following section, we develop computational methods for testing related Lyapunov conditions using SOS programming.

\subsection{Sufficient Lyapunov Conditions for Rational Rate Performance}

To show that a function $V$ satisfying the conditions in~\eqref{eq:dV_bound} certifies rational stability with a prescribed rate $k$, we first establish a standard sufficient condition in terms of the evolution of $V$ along solutions.

\begin{lem}\label{lem:LF_sufficient}
	For $f:\Omega\to\R^{n}$, let $G\subseteq\Omega$ be forward invariant. If there exists a function $V:G\to\R_{+}$ such that
	\begin{align}\label{eq:V_sufficient}
		M^{-1}\norm{x}_2^{p}
		&\leq V(x)\leq \norm{x}_2^{p},
		&&\forall x\in G, \notag\\
		V(\phi_f(x,t))
		&\leq \frac{V(x)}{1+V(x)kt},
		&&\forall x\in G,\quad t\geq0,
	\end{align}
	then $f$ is rationally stable on $G$ with power $p$, rate $k$, and gain $M$.
\end{lem}

\begin{proof}
	Suppose there exists a function $V$ satisfying~\eqref{eq:V_sufficient}, and define $\beta_{k}(y,t):=\frac{y}{1+ykt}$ for $y,t\geq0$. For every fixed $t\geq0$, $\beta_{k}(y,t)$ is monotonically nondecreasing in $y$. Therefore, for all $x\in G$ and $t\geq0$,
	\begin{align*}
		M^{-1}\norm{\phi_f(x,t)}_2^{p}
		&\leq V(\phi_f(x,t))
		\leq \beta_{k}(V(x),t)\\
		&\leq \beta_{k}(\norm{x}_2^{p},t)
		=\frac{\norm{x}_2^{p}}
		{1+\norm{x}_2^{p}kt}.
	\end{align*}
	Multiplying by $M$ yields precisely the rational stability bound in Defn.~\ref{defn:rational_stability}.
\end{proof}

Lemma~\ref{lem:LF_sufficient} provides a sufficient condition for rational stability in terms of the evolution of $V$ along trajectories. We next use the comparison principle to replace this trajectory-wise condition with a pointwise condition on the derivative $\dot{V}(x)$. Specifically, we use the following result, which follows directly from Lem.~3.4 of~\cite{khalil2002nonlinear}.

\begin{lem}\label{ref:comparison}
	For $\Omega\subseteq\R^{n}$ and $f:\Omega\to\R^{n}$, let $G\subseteq\Omega$ be forward invariant. For $\rho:\R_{+}\to\R$, let $\beta:\R_{+}\times\R_{+}\to\R_{+}$ denote the unique solution map of the scalar differential equation defined by $\rho$, so that $\partial_t\beta(y,t)=\rho(\beta(y,t))$ and $\beta(y,0)=y$. Let $V:G\to\R_{+}$ be continuous. If
	\begin{equation*}
		\dot V(x)\leq\rho(V(x)),
		\qquad \forall x\in G,
	\end{equation*}
	where $\dot{V}(x):=D_{t}^{+}V(\phi_{f}(x,t))|_{t=0}$, then
	\begin{equation*}
		V(\phi_f(x,t))
		\leq\beta(V(x),t),
		\qquad
		\forall x\in G,\quad t\geq0.
	\end{equation*}
\end{lem}

Applying Lem.~\ref{ref:comparison} with $\rho(y)=-ky^2$, we convert the trajectory-wise condition in Lem.~\ref{lem:LF_sufficient} into a pointwise Lyapunov condition for rational stability with a prescribed rate.

\begin{thm}\label{thm:LF_sufficient}
	For $\Omega\subseteq\R^{n}$ and $f:\Omega\to\R^{n}$, let $G\subseteq\Omega$ be forward invariant. If there exists a continuous function $V:G\to\R_{+}$ such that, for all $x\in G$,
	\begin{align*}
		M^{-1}\norm{x}_2^{p}
		\leq V(x)\leq \norm{x}_2^{p},\qquad
		\dot{V}(x)
		\leq-kV(x)^2,
	\end{align*}
	where $\dot{V}(x):=D_{t}^{+}V(\phi_{f}(x,t))|_{t=0}$, then $f$ is rationally stable on $G$ with power $p$, rate $k$, and gain $M$.
\end{thm}

\begin{proof}
	Suppose there exists a function $V$ satisfying the stated conditions. Define $\beta_{k}(y,t):=\frac{y}{1+ykt}$ and $\rho(y)=-ky^2$ for $y,t\geq0$. Then, for all $y,t\geq0$,
	\begin{equation*}
		\partial_{t}\beta_{k}(y,t)=\rho(\beta_{k}(y,t)),\qquad \beta_{k}(y,0)=y.
	\end{equation*}
	Since $\dot V(x)\leq\rho(V(x))$ for all $x\in G$, Lem.~\ref{ref:comparison} gives
	\begin{equation*}
		V(\phi_f(x,t))
		\leq
		\beta_{k}(V(x),t)
		=
		\frac{V(x)}{1+V(x)kt},
		\qquad
		\forall t\geq0.
	\end{equation*}
	The result then follows directly from Lem.~\ref{lem:LF_sufficient}.
\end{proof}

Thm.~\ref{thm:LF_sufficient} shows that rational stability with a prescribed rate $k$ on an invariant set $G$ can be certified by a Lyapunov function satisfying the conditions in~\eqref{eq:dV_bound} together with the stated bounds on $V$. Equivalently, the existence of such a function certifies $k$ as a lower bound on the rational rate performance. In the following subsection, we establish a converse result showing that these conditions introduce no conservatism in the rate, in the sense that rational stability with rate $k$ on $G$ implies the existence of a function satisfying the corresponding Lyapunov conditions for any $\bar{k}<k$.

\subsection{Necessary Lyapunov Conditions for Rational Rate Performance}

Having established sufficient Lyapunov conditions for rational stability with a prescribed rate in Lem.~\ref{lem:LF_sufficient} and Thm.~\ref{thm:LF_sufficient}, we now develop corresponding converse results. We first establish a converse to Lem.~\ref{lem:LF_sufficient}, showing that rational stability with a given rate implies the existence of a function satisfying the trajectory-wise conditions in~\eqref{eq:V_sufficient}.

\begin{lem}\label{lem:LF_necessary}
	For $f:\Omega\to\R^{n}$, let $G\subseteq\Omega$ be forward invariant. If $f$ is rationally stable on $G$ with power $p$, rate $k$, and gain $M$, then there exists a function $V:G\to\R_{+}$ such that, for all $x\in G$,
	\begin{align}\label{eq:V_ineq_N}
		M^{-1}\norm{x}_2^{p}
		&\leq V(x)\leq\norm{x}_2^{p}, \notag\\
		V(\phi_f(x,t))
		&\leq\frac{V(x)}{1+V(x)kt},
		&&\forall t\geq0.
	\end{align}
\end{lem}
\begin{proof}
	Define $\beta_k(y,t):=\frac{y}{1+ykt}$ whenever $y\geq 0$ and $ykt>-1$. On its domain, $\beta_k(y,t)$ is monotonically nondecreasing in $y$ and satisfies the semigroup property $\beta_k(\beta_k(y,t),s)=\beta_k(y,t+s)$ whenever the expressions involved are well-defined. Since $f$ is rationally stable on $G$	with power $p$, rate $k$, and gain $M$,
	\begin{equation*}
		M^{-1}\norm{\phi_f(x,t)}_2^p
		\leq
		\beta_k(\norm{x}_2^p,t),
		\quad
		\forall x\in G,~t\geq0.
	\end{equation*}
	Now, define $V:G\to\R_+$ as
	\begin{equation}\label{eq:converse_V}
		V(x):=
		\sup_{t\geq0}
		\left\{\beta_k\left(
		M^{-1}\norm{\phi_f(x,t)}_2^p,-t\right)\right\}. \tag{*}
	\end{equation}
	We first verify that this expression is well-defined. Since
	\begin{equation*}
		M^{-1}\norm{\phi_f(x,t)}_2^p kt
		\leq
		\beta_{k}(\norm{x}_{2}^{p},t)kt
		=\frac{\norm{x}_2^pkt}{1+\norm{x}_2^pkt}<1,
	\end{equation*}
	for $t>0$, the expression inside the supremum in~\eqref{eq:converse_V} is well-defined.	By monotonicity and the semigroup property of $\beta_k$,
	\begin{equation*}
		\beta_k\left(M^{-1}\norm{\phi_f(x,t)}_2^p,-t\right)
		\leq\beta_k\left(\beta_k(\norm{x}_2^p,t),-t\right)
		=\norm{x}_2^p.
	\end{equation*}
	Thus, the supremum in~\eqref{eq:converse_V} is finite, and $V(x)\leq\norm{x}_2^p$. Moreover, evaluating the argument at $t=0$ gives
	\begin{equation*}
		V(x)
		\geq
		\beta_k\left(
		M^{-1}\norm{\phi_f(x,0)}_2^p,0
		\right)
		=
		M^{-1}\norm{x}_2^p.
	\end{equation*}
	It remains to establish the decay condition. For any $x\in G$ and $t,s\geq0$, the semigroup properties of $\phi_f$ and $\beta_k$ give
	\begin{align*}
		&\beta_k\left(M^{-1}\norm{\phi_f(\phi_f(x,t),s)}_2^p,-s\right)\\
		&\qquad=
		\beta_k\left(M^{-1}\norm{\phi_f(x,t+s)}_2^p,-s\right)\\
		&\qquad=
		\beta_k\left(\beta_k\left(M^{-1}\norm{\phi_f(x,t+s)}_2^p,-(t+s)\right),t\right)\\
		&\qquad\leq
		\beta_k(V(x),t),
	\end{align*}
	where the final inequality follows from the definition of $V$ and the monotonicity of $\beta_k(\cdot,t)$. Taking the supremum over $s\geq0$ yields
	$	V(\phi_f(x,t))\leq	\beta_k(V(x),t)	=\frac{V(x)}{1+V(x)kt},$
	which concludes the proof.
\end{proof}

Lem.~\ref{lem:LF_necessary} shows that the trajectory-wise sufficient conditions of Lem.~\ref{lem:LF_sufficient} are also necessary. In particular, they are satisfied by the Lyapunov function
\begin{equation}\label{eq:converse_V_explicit}
	V(x)
	=
	\sup_{t\geq0}
	\left\{	\frac{M^{-1}\norm{\phi_f(x,t)}_2^p}{1-M^{-1}\norm{\phi_f(x,t)}_2^pkt}\right\}.
\end{equation}
The trajectory-wise condition in~\eqref{eq:V_ineq_N}, however, depends explicitly on the solution map and is therefore generally unsuitable for direct computational verification. Thm.~\ref{thm:LF_sufficient} avoids this dependence by using the comparison principle to replace the trajectory-wise condition with a pointwise inequality on $\dot V$. To obtain a corresponding converse result for this differential condition, we use the following converse implication of the comparison principle.

\begin{lem}\label{lem:comparison_necessary}
	For $\Omega\subseteq\R^n$ and $f:\Omega\to\R^n$, let $G\subseteq\Omega$ be forward invariant. For $\rho:\R_+\to\R$, let $\beta:\R_+\times\R_+\to\R_+$ denote the unique solution map of the ODE defined by $\rho$, so that $\beta(y,0)=y$ and $\partial_t\beta(y,t)=\rho(\beta(y,t))$.
	Let $V:G\to\R_+$ be continuous. If
	\begin{equation*}
		V(\phi_f(x,t))
		\leq
		\beta(V(x),t),
		\qquad
		\forall x\in G,\quad t\geq0,
	\end{equation*}
	then
	\begin{equation*}
		\dot V(x)\leq\rho(V(x)),
		\qquad
		\forall x\in G,
	\end{equation*}
	where $\dot{V}(x):=D_{t}^{+}V(\phi_{f}(x,t))|_{t=0}$.
\end{lem}

\begin{proof}
	For any $x\in G$, the assumed inequality and the identity $V(\phi_f(x,0))=V(x)=\beta(V(x),0)$ imply
	\begin{align*}
		\dot V(x)
		&=
		\limsup_{\Delta t\to0^+}
		\frac{
			V(\phi_f(x,\Delta t))-V(\phi_f(x,0))}
		{\Delta t}\\
		&\leq
		\limsup_{\Delta t\to0^+}
		\frac{
			\beta(V(x),\Delta t)-\beta(V(x),0)}
		{\Delta t}
		=
		\rho(V(x)).
	\end{align*}
\end{proof}

To apply Lem.~\ref{lem:comparison_necessary}, the Lyapunov function must be continuous. Continuity of the function in~\eqref{eq:converse_V_explicit} is not immediate because its definition involves a supremum over an unbounded time interval. We can resolve this by introducing an arbitrarily small reduction in the certified rate, yielding the following converse Lyapunov theorem.



\begin{thm}\label{thm:LF_necessary}
	For $\Omega\subseteq\R^n$ and $f:\Omega\to\R^n$, let $G\subseteq\Omega$ be forward invariant. Suppose that $f$ is rationally stable on $G$ with power $p$, rate $k>0$, and gain $M$. Then, for every $\bar{k}\in[0,k)$, there exists a continuous function $V:G\to\R_+$ such that, for all $x\in G$,
	\begin{equation*}
		M^{-1}\norm{x}_2^p
		\leq V(x)\leq\norm{x}_2^p,\qquad
		\dot V(x)
		\leq-\bar{k}V(x)^2,
	\end{equation*}
	where $\dot{V}(x):=D_{t}^{+}V(\phi_{f}(x,t))|_{t=0}$.
\end{thm}

\begin{proof}
	Fix any $\bar{k}\in[0,k)$ and define $\beta_{\bar{k}}(y,t):=\frac{y}{1+y\bar{k}t}$ whenever
	$y\geq0$ and $y\bar{k}t>-1$. Consider the function
	\begin{equation*}
		V(x):=\sup_{t\geq0}\left\{\beta_{\bar{k}}\left(M^{-1}\norm{\phi_f(x,t)}_2^p,-t\right)\right\}.
	\end{equation*}
	We first show that $V$ is well-defined. Since $f$ is rationally stable with rate $k$ and	$\frac{\norm{x}_2^p\bar{k}t}{1+\norm{x}_2^pkt}<1$, both terms below are well-defined and, by monotonicity of $\beta_{\bar{k}}$,
	\begin{align}\label{eq:converse_tail_bound}
		\beta_{\bar{k}}\left(
		M^{-1}\norm{\phi_f(x,t)}_2^p,-t
		\right)
		&\leq \beta_{\bar{k}}\left(\frac{\norm{x}_2^p}{1+\norm{x}_2^pkt},-t\right) \notag\\
		&=\frac{\norm{x}_2^p}{1+\norm{x}_2^p(k-\bar{k})t}. \tag{$\star$}
	\end{align}
	It follows that $V$ is well-defined and bounded above by $\norm{x}_2^p$. Evaluating at $t=0$ also gives $V(x)\geq M^{-1}\norm{x}_2^p$, establishing the upper and lower bounds on $V$.

	We next establish continuity of $V$. Define
	\begin{equation*}
		h(x,t)
		:=
		\beta_{\bar{k}}\left(
		M^{-1}\norm{\phi_f(x,t)}_2^p,-t
		\right),
	\end{equation*}
	so that $V(x)=\sup_{t\geq 0}h(x,t)$.
	For every finite $T>0$, $h$ is continuous on $G\times[0,T]$. Moreover,~\eqref{eq:converse_tail_bound} gives
	\begin{equation*}
		0\leq h(x,t)
		\leq
		\frac{\norm{x}_2^p}
		{1+\norm{x}_2^p(k-\bar{k})t}
		\leq
		\frac{1}{(k-\bar{k})t},
		\quad t>0.
	\end{equation*}
	Thus, $h(x,t)\to0$ as $t\to\infty$ uniformly in $x\in G$. Consequently, for any $\epsilon>0$, there exists $T_\epsilon>0$ such that $t\geq T_\epsilon$ implies $h(x,t)<\epsilon$ for all $x\in G$. Defining
	\begin{equation*}
		V_{T_\epsilon}(x)
		:=
		\max_{t\in[0,T_\epsilon]}h(x,t),
	\end{equation*}
	it follows that $V(x)\leq \max\{V_{T_{\epsilon}}(x),\epsilon\}$ and therefore $0\leq V(x)- V_{T_{\epsilon}}(x)\leq\epsilon$ for all $x\in G$. Hence $V_{T_\epsilon}\to V$ uniformly on $G$ as $\epsilon\to0$. Since $h$ is continuous and $[0,T_\epsilon]$ is compact, $V_{T_\epsilon}$ is continuous, and thus $V$ is continuous as well.
	
	Finally, we establish the dissipation inequality. By a similar semigroup argument as in the proof of Lem.~\ref{lem:LF_necessary}, we have
	\begin{equation*}
		V(\phi_f(x,t))
		\leq
		\beta_{\bar{k}}(V(x),t)
		=
		\frac{V(x)}
		{1+V(x)\bar{k}t}.
	\end{equation*}
	Since $\partial_t\beta_{\bar{k}}(y,t)=-\bar{k}\beta_{\bar{k}}(y,t)^2$, Lem.~\ref{lem:comparison_necessary} then gives
	\begin{equation*}
		\dot V(x)\leq-\bar{k}V(x)^2,
		\qquad
		\forall x\in G,
	\end{equation*}
	which completes the proof.
\end{proof}

Thms.~\ref{thm:LF_sufficient} and~\ref{thm:LF_necessary} together provide an exact Lyapunov characterization of rational rate performance. Specifically, rational stability with rate $k$ implies the existence of a continuous Lyapunov function satisfying the conditions of Thm.~\ref{thm:LF_sufficient} for every $\bar{k}<k$, while feasibility of these conditions at any rate $\bar{k}$ certifies rational stability with that same rate. Consequently, the supremum of the rates certified by these Lyapunov conditions coincides with the rational rate performance. Moreover, if $V$ is differentiable, its upper Dini derivative along solutions reduces to $\dot V(x)=\nabla V(x)^Tf(x)$, allowing the differential Lyapunov condition to be verified without explicit knowledge of the solution map $\phi_f$. In the following section, we develop tractable relaxations of these conditions that can be enforced using SOS programming.

\section{SOS Conditions for Rational Stability Rate Analysis}\label{sec:SOS_conditions}

Having established necessary and sufficient Lyapunov conditions for rational stability on a prescribed domain, we now develop numerical conditions for verifying rational stability for polynomial vector fields as a Sum-of-Squares Program (SOSP). The main challenge is that the derivative condition in Thm.~\ref{thm:LF_sufficient}, $\dot{V}(x)\leq-kV(x)^2$, is nonlinear in the unknown function $V$ and therefore cannot be imposed directly as a convex SOS constraint. Since $V(x)\leq\norm{x}_2^p$, however, this condition can be strengthened to $\dot{V}(x)\leq-kV(x)\norm{x}_2^p$, which is affine in $V$ for fixed $k$. Strengthening it further to $\dot{V}(x)\leq-k\norm{x}_2^{2p}$ recovers a more classical form of the Lyapunov conditions for rational stability; see, e.g.,~\cite{bacciotti2005LF_book}. The following lemma formalizes the relationships among these three classes of Lyapunov conditions.

\begin{lem}\label{lem:rational_LFs_equivalence}
	For $\Omega\subseteq\R^{n}$ and $f:\Omega\to\R^{n}$, let
	$G\subseteq\Omega$ be forward invariant. Consider the following conditions:
	\begin{enumerate}
		\item[\tnf{(i)}]
		There exist $C_1,C_2>0$, $C_3\geq0$, $q>r>0$, and a
		continuous function $V_{\tnf{(i)}}:G\to\R_+$ such that
		\begin{align*}
			C_1\norm{x}_2^r
			\leq V_{\tnf{(i)}}(x)
			&\leq C_2\norm{x}_2^r,\\
			\dot V_{\tnf{(i)}}(x)
			&\leq-C_3\norm{x}_2^q,
			\qquad \forall x\in G.
		\end{align*}
		
		\item[\tnf{(ii)}]
		There exist $c\geq0$, $\gamma,p,r>0$, and a continuous
		function $V_{\tnf{(ii)}}:G\to\R_+$ such that
		\begin{align*}
			\gamma^{-1}\norm{x}_2^r
			\leq V_{\tnf{(ii)}}(x)
			&\leq\norm{x}_2^r,\\
			\dot V_{\tnf{(ii)}}(x)
			&\leq-cV_{\tnf{(ii)}}(x)\norm{x}_2^p,
			\qquad \forall x\in G.
		\end{align*}
		
		\item[\tnf{(iii)}]
		There exist $c_1,c_2>0$, $c_3\geq0$, $r,s>0$, and a
		continuous function $V_{\tnf{(iii)}}:G\to\R_+$ such that
		\begin{align*}
			c_1\norm{x}_2^r
			\leq V_{\tnf{(iii)}}(x)
			&\leq c_2\norm{x}_2^r,\\
			\dot V_{\tnf{(iii)}}(x)
			&\leq-c_3V_{\tnf{(iii)}}(x)^{1+s},
			\qquad \forall x\in G.
		\end{align*}
		
		\item[\tnf{(iv)}]
		$f$ is rationally stable on $G$ with power $p>0$, rate $k>0$,
		and gain $M\geq 1$.
	\end{enumerate}
	Then, the following statements hold:
	\begin{enumerate}
		\item \tnf{(i)} implies \tnf{(iv)} with $p=q-r$, $M=(\frac{C_{2}}{C_{1}})^{\frac{q}{r}-1}$, $k=(\frac{q}{r}-1)\frac{C_{3}}{C_{2}}$;		
		
		\item \tnf{(iv)} implies \tnf{(i)} with $C_{1}=\frac{1}{M}$, $C_{2}=1$, $C_{3}=\frac{\bar{k}}{M^2}$, $r=p$, $q=2p$, for any $\bar{k}\in[0,k)$;
		
		\item \tnf{(ii)} implies \tnf{(iv)} with $M=\gamma^{p/r}$, $k=\frac{cp}{r}$;
		
		\item \tnf{(iv)} implies \tnf{(ii)} with $r=p$, $\gamma=M$, $c=\frac{\bar{k}}{M}$, for any $\bar{k}\in[0,k)$;
		
		\item \tnf{(iii)} implies \tnf{(iv)} with $p=rs$, $M=(\frac{c_2}{c_1})^s$, $k=s\frac{c_3}{c_{2}^{s}}$;
		
		\item \tnf{(iv)} implies \tnf{(iii)} with $c_{1}=\frac{1}{M}$, $c_{2}=1$, $c_{3}=\bar{k}$, $r=p$, and $s=1$, for any $\bar{k}\in[0,k)$.
	\end{enumerate}
\end{lem}

\begin{proof}
	Statements 1), 3), and 5) follow from Thm.~\ref{thm:LF_sufficient} using the Lyapunov functions $V(x)=(C_{2}^{-1}V_{\tnf{(i)}}(x))^{p/r}$, $V(x)=V_{\tnf{(ii)}}(x)^{p/r}$, and $V(x)=(c_{2}^{-1}V_{\tnf{(iii)}}(x))^{s}$, respectively. Statements 2), 4), and 6) follow from Thm.~\ref{thm:LF_necessary} by taking $V_{\tnf{(i)}}(x)=V_{\tnf{(ii)}}(x)=V_{\tnf{(iii)}}(x)=V(x)$, where $V$ is the converse Lyapunov function from the proof of Thm.~\ref{thm:LF_necessary}. A full proof is provided in Appx.~\ref{appx:Proofs}.
\end{proof}

Lem.~\ref{lem:rational_LFs_equivalence} presents three equivalent classes of Lyapunov conditions for rational stability, in the sense that feasibility of any one implies rational stability for suitable parameters $p$, $k$, and $M$, and rational stability conversely implies feasibility of each class. Conditions~\tnf{(i)} correspond to the classical conditions for rational stability; see, e.g.,~\cite{bacciotti2005LF_book}. For polynomial $f$ on bounded regions, feasibility of these conditions implies feasibility for a polynomial Lyapunov function, making them particularly amenable to SOS programming. However, Conditions~\tnf{(i)} need not preserve prescribed performance parameters. In particular, fixing $r=p$, $q=2p$, and $C_{1}=M^{-1}$, feasibility of \tnf{(i)} with given $C_{2}$ and $C_{3}$ implies rational stability with rate $k=\frac{C_{3}}{C_{2}}$, whereas rational stability with rate $k$ guarantees only $\frac{C_{3}}{C_{2}}\in[0,\frac{k}{M^2})$ through the converse construction. Thus, Conditions~\tnf{(i)} may deviate by a factor of $1/M^2$ when used to certify rational rate performance.

On the other hand, Conditions~\tnf{(iii)} are similar to those used in, e.g.,~\cite{jammazi2013rational} and provide a performance-exact characterization. In particular, taking $c_{1}=M^{-1}$, $c_{2}=1$, $r=p$, $s=1$, and $c_{3}=k$ recovers the conditions of Thm.~\ref{thm:LF_sufficient}. However, the resulting derivative inequality is nonlinear in the unknown $V$, preventing direct formulation as an SOS program.

Conditions~\tnf{(ii)} provide an intermediate characterization that trades some performance conservatism for computational tractability. Statements 3) and 4) show that these conditions may lose a factor of $1/M$ when certifying rational rate performance, improving on the factor of $1/M^2$ associated with Conditions~\tnf{(i)}. At the same time, Conditions~\tnf{(ii)} remain affine in $V$ and can therefore be imposed using convex optimization. They thus provide a potentially less conservative alternative to the classical conditions while remaining suitable for SOS programming for polynomial systems.

To formulate SOS conditions corresponding to Conditions~\tnf{(ii)}, let $\Sigma_{s,d}$ denote the set of SOS polynomials of degree at most $2d$, so that $s\in\Sigma_{s,d}$ implies $s(x)=Z_d(x)^TPZ_d(x)$ for some $P\succeq0$, where $Z_d$ is the vector of monomials of degree at most $d$. For a semialgebraic region $\Omega:=\{x\in\R^n\mid g_i(x)\geq0,\ i=1,\ldots,m\}$, with $g_i\in\R_{d_i}[x]$, we use $s\in\Sigma[\Omega]$ to denote an SOS certificate of nonnegativity on $\Omega$, motivated by Putinar's Positivstellensatz~\cite{putinar1993psatz}. In particular, $s\in\Sigma[\Omega]$ implies $s(x)=s_0(x)+\sum_i s_i(x)g_i(x)$, where $s_0$ and $s_i$ are SOS polynomials of appropriate prescribed degrees.

\begin{cor}\label{cor:rational_stability_SOS}
	Let $d,d',r,p\in\N$ with $r,p$ even, $f\in\R_{d_f}^n[x]$, and $g_i\in\R_{d_i}[x]$, and define $\Omega:=\{x\in\R^n\mid g_i(x)\geq0,\ i=1,\ldots,m\}$. For a prescribed $k\geq0$, suppose there exist $\gamma>0$ and $V\in\R_{2d}[x]$ satisfying
	\begin{align}\label{eq:rational_stability_SOSP}
		V(x)-(x^Tx)^{r/2}
		&\in\Sigma_{s,d}[\Omega],\notag\\
		\gamma(x^Tx)^{r/2}-V(x)
		&\in\Sigma_{s,d}[\Omega],\notag\\[-0.4em]
		-\frac{kr}{p}V(x)(x^Tx)^{p/2}-\nabla V(x)^Tf(x)
		&\in\Sigma_{s,d'}[\Omega].
	\end{align}
	Then, for any $c\geq0$ such that $S_c(V)\subseteq\tnf{int}(\Omega)$ where $S_c(V):=\{x\in\R^n\mid V(x)\leq c\}$, $f$ is rationally stable on $S_c(V)$ with power $p$, rate $k$, and gain $M:=\gamma^{p/r}$.
\end{cor}

\begin{proof}
	Suppose $V$ satisfies~\eqref{eq:rational_stability_SOSP} and let $M:=\gamma^{p/r}$. Then $V(x)\leq\gamma\norm{x}_2^r$ and $\nabla V(x)^Tf(x)\leq 0$ for all $x\in\Omega$. Thus, any sublevel set $S_c(V)\subseteq\operatorname{int}(\Omega)$ is forward invariant.
	
	Define $\tilde{V}:=\gamma^{-1}V$. Then, for all $x\in S_c(V)$, 
	\begin{align*}
		\gamma^{-1}\norm{x}_2^r\leq\tilde V(x)&\leq\norm{x}_2^r,	\\
		\nabla\tilde V(x)^Tf(x)&\leq-\frac{kr}{p}\tilde V(x)\norm{x}_2^p,
	\end{align*}
	Thus, $\tilde{V}$ satisfies Condition~\tnf{(ii)} of Lem.~\ref{lem:rational_LFs_equivalence} with $c=kr/p$. By Statement 3) of Lem.~\ref{lem:rational_LFs_equivalence}, there exists $\hat V:S_c(V)\to\R_+$ satisfying $M^{-1}\norm{x}_2^p\leq\hat V(x)\leq\norm{x}_2^p$ and $\nabla\hat V(x)^Tf(x)\leq-k\hat V(x)^2$. Thm.~\ref{thm:LF_sufficient} therefore implies that $f$ is rationally stable on $S_c(V)$ with power $p$, rate $k$, and gain $M$.
\end{proof}

For fixed $k$, the conditions in~\eqref{eq:rational_stability_SOSP} are convex SOS constraints in $V$, $\gamma$, and the SOS multipliers. Moreover, feasibility is monotone in $k$, so the largest certifiable rate can be computed to arbitrary accuracy by bisection over $k$, with a convex SOS problem solved at each iteration.

In the following section, we apply these SOS conditions to several numerical examples to estimate rational rate performance and compute regions on which prescribed levels of performance can be guaranteed.

\section{Numerical Examples}\label{sec:examples}

In this section, we illustrate the proposed Lyapunov conditions by estimating global and local rational rate performance for several nonlinear ODEs. We primarily use Conditions~\tnf{(ii)} of Lem.~\ref{lem:rational_LFs_equivalence}, implemented through the SOSP in Cor.~\ref{cor:rational_stability_SOS}. To assess their conservatism, the resulting lower bounds on rate performance are compared with simulation-based estimates and with bounds obtained using the classical Lyapunov Conditions~\tnf{(i)}. The latter are converted to SOS constraints analogously to Conditions~\tnf{(ii)}, yielding the derivative constraint
\begin{equation}\label{eq:rational_SOS_classical}
	-\frac{kr}{p}M^{\frac{r}{p}} (x^Tx)^{(r+p)/2}-\nabla V(x)^T f
	\in \Sigma_{s,d'}[\Omega].
\end{equation}
For each comparison, the same values of $M$, $p$, $r$, $d$, and $d'$ are
used for both sets of Lyapunov conditions. Bisection over $k$ is used to
compute the greatest lower bound on rate performance for a prescribed gain
$M$. For the first two examples, bisection over $M$ is also used to compute
the smallest gain for which a fixed rate $k$ is feasible. The SOS programs
are parsed using SOSTOOLS~\cite{papachristodoulou2021SOSTOOLS}, and the
resulting semidefinite programs are solved using MOSEK~\cite{mosek}.

\subsection{Conservatism Scaling in Rate Performance}
\label{subsec:examples:conservatism}

We first consider a simple example for which the rational rate performance
and the conservatism of the different Lyapunov conditions can be determined
explicitly:
\begin{equation}\label{eq:ODE_conservatism}
	\dot{x}(t)=f(x(t))
	=\bmat{-\frac{1}{4}x_{1}^{3}-\frac{1}{2}x_{1}x_{2}^{2}-2x_{2}\\
		-\frac{1}{4}x_{1}^{2}x_{2}-\frac{1}{2}x_{2}^{3}+x_{1}}.
\end{equation}
The vector field $f$ is globally rationally stable with power $p=2$, rate
$k=1$, and gain $M=2$. Indeed, the Lyapunov function
$V(x)=\frac{1}{2}x_{1}^{2}+x_{2}^{2}$ satisfies
$\frac{1}{2}\norm{x}_{2}^{2}\leq V(x)\leq\norm{x}_{2}^{2}$ and
\begin{equation*}
	\nabla V(x)^T f(x)
	=-V(x)x_{1}^2-2V(x)x_{2}^2=-V(x)^2.
\end{equation*}
Thus, by Statement 5) of Lem.~\ref{lem:rational_LFs_equivalence}, the system
is rationally stable with $k_{\tnf{(iii)}}=1$ and $M=2$. As shown in Fig.~\ref{fig:conservatism_scaling}, the bound is attained along the trajectory from $x(0)=\tnf{e}_{2}:=(0,1)$, so the corresponding rate performance is tight.

Applying the SOSP in~\eqref{eq:rational_stability_SOSP} with $p=r=2$ and $d=1$ recovers the same Lyapunov function and certifies rational stability with rate $k_{\tnf{(ii)}}=0.5$ and gain $M=2$, exactly matching the predicted $1/M=\frac{1}{2}$ reduction from the true rate. Increasing the degree of the candidate Lyapunov function does not improve this rate. Applying the classical Conditions~\tnf{(i)} similarly yields $k_{\tnf{(i)}}=0.25$, corresponding exactly to the predicted factor $1/M^2=1/4$, and this bound likewise does not improve with higher-degree candidate Lyapunov functions. Fig.~\ref{fig:conservatism_scaling} shows the resulting rational stability bounds $\beta_{\zeta}(1,t)=\frac{2}{1+k_{\zeta}t}$ along with simulated solutions for $x(0)=\tnf{e}_{1}:=(1,0)$ and $x(0)=\tnf{e}_{2}:=(0,1)$.

\begin{figure}[t]
	\centering
	\includegraphics[width=1.0\linewidth]{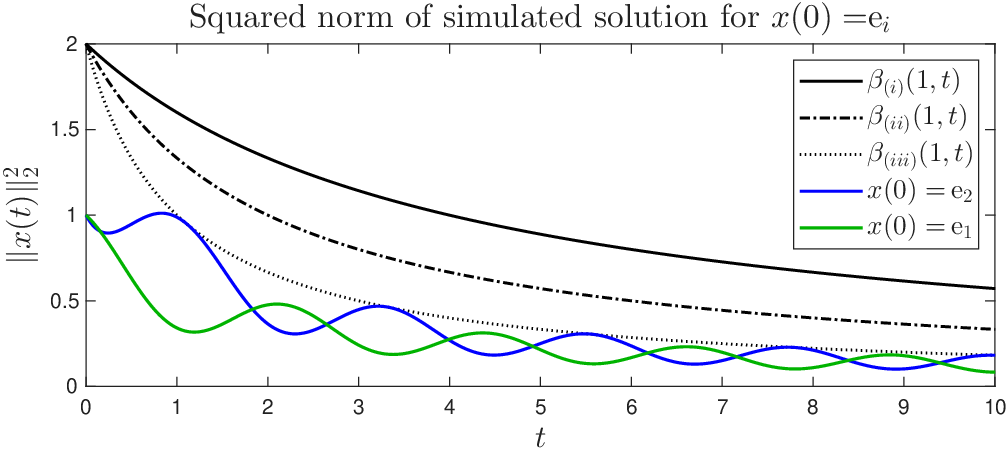}
	\vspace*{-0.6cm}
	\caption{Squared norms of simulated solutions to~\eqref{eq:ODE_conservatism}
		for initial states $x(0)=\tnf{e}_{1}$ and $x(0)=\tnf{e}_{2}$, together
		with the rational stability bounds
		$\beta_{\zeta}(1,t)=\frac{2}{1+k_{\zeta}t}$ for rates
		$k_{\tnf{(i)}}=0.25$, $k_{\tnf{(ii)}}=0.5$, and
		$k_{\tnf{(iii)}}=1$.}
	\label{fig:conservatism_scaling}
	\vspace*{-0.4cm}
\end{figure}

\subsection{Global Rational Stability Rate Performance}
\label{subsec:examples:global}

We next consider the three-dimensional polynomial ODE
\begin{equation}\label{eq:example1}
	\dot{x}=f(x):=\bmat{-10x_{2}^3\\
		10x_{1}^3-5x_{3}^2x_{2}\\
		-10x_{3}^3 +10x_{2}^2x_{1}}.
\end{equation}
To estimate its global rational rate performance, we solve
SOSP~\eqref{eq:rational_stability_SOSP} with $p=2$,
$d\in\{2,\ldots,9\}$, and $r=2d$. The resulting greatest lower bounds
$k_{\tnf{(ii)}}$ are presented in Tab.~\ref{tab:rational_stability_3D},
together with the associated smallest gains $M$. As the degree of the
candidate Lyapunov function increases, the certified rate initially improves
and converges to approximately $k_{\tnf{(ii)}}=0.825$, while the gain
generally decreases toward $M\approx2.8$.

To assess the conservatism of the SOS certificates, we also estimate the rate
performance through simulation. Simulating solutions up to $t=100$ from
5000 initial conditions distributed on the unit sphere, we obtain an estimated
rate $k_{\tnf{sim}}=1.7596$ and associated estimated gain
$M_{\tnf{sim}}=2.748$ consistent with the simulated trajectories.
Figure~\ref{fig:rational_bound_global} shows solutions for initial states
$x(0)=\tnf{u}_i$, where
\begin{equation}\label{eq:global_x0}
	\tnf{u}_{1}:=\bmat{1\\0\\0},\quad
	\tnf{u}_{2}:=\bmat{0\\1\\0},\quad
	\tnf{u}_{3}:=\bmat{-0.93\\ 0.06\\ 0.37},
\end{equation}
with $\tnf{u}_{3}$ corresponding to the sampled trajectory that determines the
simulation-based rate estimate. The rational stability bounds associated
with $k_{\tnf{sim}}$ and $k_{\tnf{(ii)}}$ are also shown, using
$M=M_{\tnf{sim}}$. The ratio
$k_{\tnf{(ii)}}/k_{\tnf{sim}}\approx0.47$ is consistent with the predicted
$1/M$ scaling, for which $1/M_{\tnf{sim}}\approx0.36$.

Using the classical Conditions~\tnf{(i)} with the same polynomial degrees and corresponding values of $M$ yields the bounds $k_{\tnf{(i)}}$ in Tab.~\ref{tab:rational_stability_3D}. The largest certified rate is $k_{\tnf{(i)}}=0.191$, substantially below $k_{\tnf{(ii)}}=0.825$. Moreover, unlike Conditions~\tnf{(ii)}, increasing the degree of the candidate Lyapunov function results in a worse bound on the rate performance. This behavior is likely a result of the additional restrictions imposed by the global SOS implementation of the classical conditions, requiring the Lyapunov function to be a homogeneous polynomial for global certification. Overall, Conditions~\tnf{(ii)} provide substantially tighter estimates of global rational rate performance in this example.

\begin{figure}[t]
	\centering
	\includegraphics[width=1.0\linewidth]{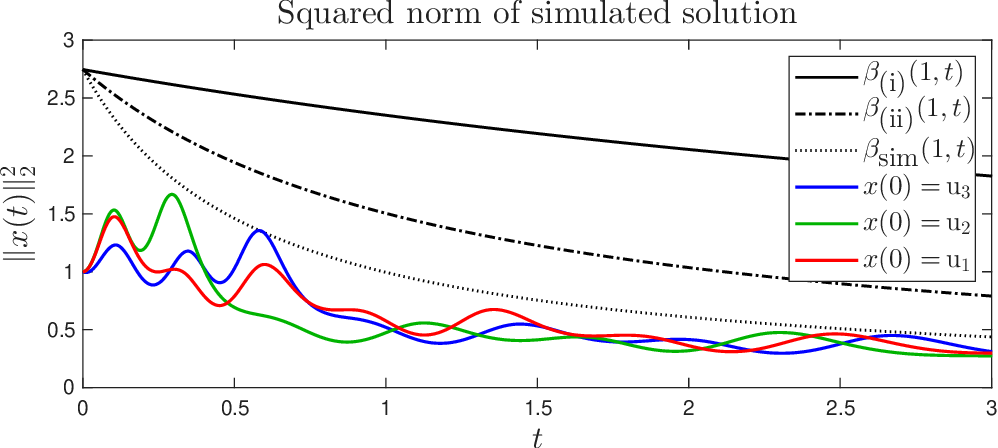}
	\vspace*{-0.6cm}
	\caption{Squared norms of simulated solutions to~\eqref{eq:example1} for
		initial states $x(0)=\tnf{u}_{i}$ from~\eqref{eq:global_x0}, together
		with the rational stability bounds
		$\beta_{\zeta}(1,t)=\frac{2.75}{1+k_{\zeta}t}$ for rates
		$k_{\tnf{(i)}}=0.191$, $k_{\tnf{(ii)}}=0.825$, and
		$k_{\tnf{sim}}=1.76$.}
	\label{fig:rational_bound_global}
	\vspace*{-0.4cm}
\end{figure}

\begin{table}[b]
	\setlength{\tabcolsep}{4pt}
	\begin{tabular}{c|cccccccc}
		$d$ & $2$ & $3$ & $4$ & $5$ & $6$ & $7$ & $8$ & $9$ \\\hline
		$k_{\tnf{(i)}}$ & 0.191 & 0.151 & 0.110 & 0.0617 & 0.0351 &
		0.0201 & 0.0117 & 0.0068\\
		$k_{\tnf{(ii)}}$ & 0.493 & 0.683 & 0.805 & 0.825 & 0.825 &
		0.825 & 0.825 & 0.825\\
		$M$ & 9.108 & 5.141 & 4.407 & 5.814 & 3.304 & 3.031 & 2.829 & 2.793
	\end{tabular}
	\caption{Greatest lower bounds on the global rational rate performance
		of~\eqref{eq:example1} obtained using Lyapunov Conditions~\tnf{(i)}
		($k_{\tnf{(i)}}$) and~\tnf{(ii)} ($k_{\tnf{(ii)}}$), together with
		the smallest gain $M$ for which rate $k_{\tnf{(ii)}}$ can be certified.
		The SOSPs use $r=2d$, $p=2$.}
	\label{tab:rational_stability_3D}
	\vspace*{-0.4cm}
\end{table}

\subsection{Regions of Rational Stability Rate Performance}

\begin{table}[b]
	\setlength{\tabcolsep}{5pt}
	\begin{tabular}{c|cccccccc}
		$R$ & 0.025 & 0.200 & 0.350 & 0.475 & 0.575 & 0.650 & 0.700 & 0.725 \\\hline
		$k_{\tnf{(i)}}$ & 0.211 & 0.212 & 0.177 & 0.176 & 0.169 & 0.153 &
		0.128 & 0.110\\
		$k_{\tnf{(ii)}}$ & 1.203 & 1.152 & 1.056 & 0.939 & 0.697 & 0.442 &
		0.267 & 0.177\\
		$k_{\tnf{sim}}$ & 2.443 & 1.423 & 1.328 & 1.204 & 1.057 & 0.887 &
		0.748 & 0.667
	\end{tabular}
	\caption{Bounds $k_{\tnf{(ii)}}$ on the rational rate performance of~\eqref{eq:example2} over the regions in Fig.~\ref{fig:rationa_stability_ROA}, computed as the largest Lyapunov sublevel sets obtained from SOSP~\eqref{eq:rational_stability_SOSP}	on $\Omega:=\{x\in\R^{2}\mid x_1^4+x_2^4\leq R^{4}\}$ with $p=2$ and $M=3$. Also shown are bounds from Conditions~\tnf{(i)} ($k_{\tnf{(i)}}$) and simulation-based estimates ($k_{\tnf{sim}}$).}
	\label{tab:rational_stability_local}
\end{table}

Finally, consider the locally stable polynomial ODE
\begin{equation}\label{eq:example2}
	\dot{x}=f(x):=\bmat{
		2x_{1}^5 + 2x_{1}^3x_{2}^2 -x_{1}^3 - 8x_{2}^3\\
		x_{1}^4x_{2}^3 + x_{2}^7 + x_{1}^3 - x_{2}^3}.
\end{equation}
We estimate regional rational rate performance by solving SOSP~\eqref{eq:rational_stability_SOSP} on $\Omega:=\{x\in\R^{2}\mid x_1^4+x_2^4\leq R^{4}\}$ for several values of $R$, using $p=2$ and $M=3$. The resulting greatest lower bounds $k_{\tnf{(ii)}}$ are presented in Tab.~\ref{tab:rational_stability_local}. The largest closed level sets of the corresponding Lyapunov functions are shown in Fig.~\ref{fig:rationa_stability_ROA}, providing guaranteed regions of rational stability with gain $M=3$ and rate $k=k_{\tnf{(ii)}}$. A simulation-based estimate of the region of attraction is also shown, together with the Lyapunov level set for $k=0$. As the certified rate decreases, the corresponding region of guaranteed rational performance expands toward the region of attraction.

Simulation-based rate estimates are also reported in Tab.~\ref{tab:rational_stability_local}, obtained from 1000 initial conditions distributed along the boundary of each region. Bounds $k_{\tnf{(i)}}$ from the classical Conditions~\tnf{(i)} are provided for comparison. Across all tested regions, Conditions~\tnf{(ii)} certify substantially larger rates than Conditions~\tnf{(i)}, although the gap between the SOS-certified and simulation-based rates generally increases with the region size. The ratio $k_{\tnf{(ii)}}/k_{\tnf{sim}}$ generally decreases as the region grows, but appears broadly consistent with the predicted $1/M=\frac{1}{3}$ scaling.

\begin{figure}[t]
	\centering
	\includegraphics[width=1.0\linewidth]{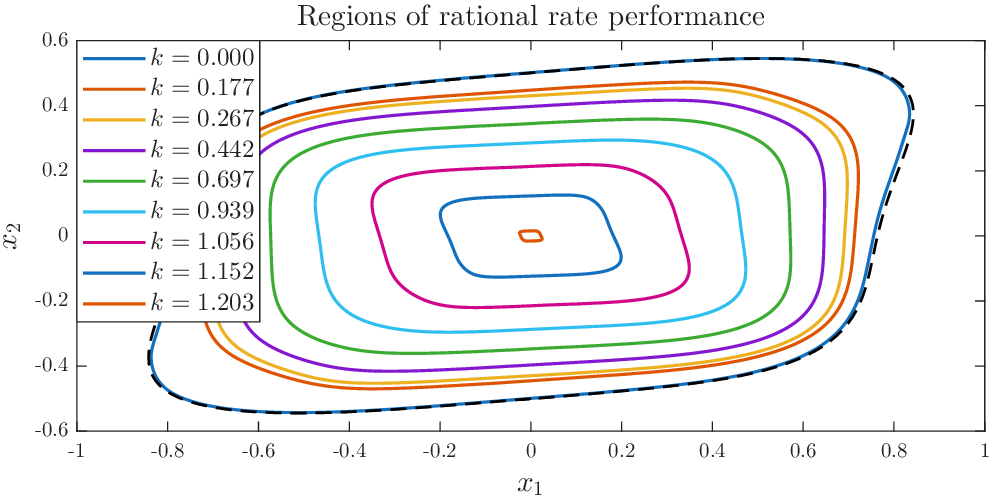}
	\vspace*{-0.6cm}
	\caption{Forward-invariant regions of~\eqref{eq:example2} on which
		rational stability is certified with different rates
		$k=k_{\tnf{(ii)}}$ using SOSP~\eqref{eq:rational_stability_SOSP}.
		A simulation-based estimate of the region of attraction is shown by
		the black dashed curve.}
	\label{fig:rationa_stability_ROA}
	\vspace*{-0.4cm}
\end{figure}

\section{Conclusion}

In this paper, Lyapunov-based methods were developed for quantifying the rate performance of rationally stable nonlinear ODEs. A notion of rational rate performance on prescribed invariant sets was introduced, and an exact converse Lyapunov characterization of this performance was established. In particular, the parameters of a prescribed rational solution bound can be recovered up to an arbitrarily small relaxation of the rate. A tiered set of relaxations was then developed, including a convex condition that is affine in the unknown Lyapunov function and appears less conservative than classical norm-based conditions. For polynomial systems, these conditions were formulated as SOS programs for estimating rational rate performance and certifying prescribed performance on invariant regions. Finally, the predicted conservatism was examined numerically, demonstrating improvements over the classical SOS-based approach and computing nested regions of state-space on which prescribed levels of performance can be guaranteed.


\bibliographystyle{IEEEtran}
\bibliography{bibfile}

\clearpage

\begin{appendices}

\section{Proof of Lemma~\ref{lem:rational_LFs_equivalence}}\label{appx:Proofs}

In this appendix, we restate and prove Lem.~\ref{lem:rational_LFs_equivalence}, establishing three different characterizations of rational stability. In particular, we have the following.

\indent\textit{Lemma~\ref{lem:rational_LFs_equivalence}:}
For $\Omega\subseteq\R^{n}$ and $f:\Omega\to\R^{n}$, let
$G\subseteq\Omega$ be forward invariant. Consider the following conditions:
\begin{enumerate}
	\item[\tnf{(i)}]
	There exist $C_1,C_2>0$, $C_3\geq0$, $q>r>0$, and a
	continuous function $V_{\tnf{(i)}}:G\to\R_+$ such that
	\begin{align*}
		C_1\norm{x}_2^r
		\leq V_{\tnf{(i)}}(x)
		&\leq C_2\norm{x}_2^r,\\
		\dot V_{\tnf{(i)}}(x)
		&\leq-C_3\norm{x}_2^q,
		\qquad \forall x\in G.
	\end{align*}
	
	\item[\tnf{(ii)}]
	There exist $c\geq0$, $\gamma,p,r>0$, and a continuous
	function $V_{\tnf{(ii)}}:G\to\R_+$ such that
	\begin{align*}
		\gamma^{-1}\norm{x}_2^r
		\leq V_{\tnf{(ii)}}(x)
		&\leq\norm{x}_2^r,\\
		\dot V_{\tnf{(ii)}}(x)
		&\leq-cV_{\tnf{(ii)}}(x)\norm{x}_2^p,
		\qquad \forall x\in G.
	\end{align*}
	
	\item[\tnf{(iii)}]
	There exist $c_1,c_2>0$, $c_3\geq0$, $r,s>0$, and a
	continuous function $V_{\tnf{(iii)}}:G\to\R_+$ such that
	\begin{align*}
		c_1\norm{x}_2^r
		\leq V_{\tnf{(iii)}}(x)
		&\leq c_2\norm{x}_2^r,\\
		\dot V_{\tnf{(iii)}}(x)
		&\leq-c_3V_{\tnf{(iii)}}(x)^{1+s},
		\qquad \forall x\in G.
	\end{align*}
	
	\item[\tnf{(iv)}]
	$f$ is rationally stable on $G$ with power $p>0$, rate $k>0$,
	and gain $M\geq 1$.
\end{enumerate}
Then, the following statements hold:
\begin{enumerate}
	\item \tnf{(i)} implies \tnf{(iv)} with $p=q-r$, $M=(\frac{C_{2}}{C_{1}})^{\frac{q}{r}-1}$, $k=(\frac{q}{r}-1)\frac{C_{3}}{C_{2}}$;		
	
	\item \tnf{(iv)} implies \tnf{(i)} with $C_{1}=\frac{1}{M}$, $C_{2}=1$, $C_{3}=\frac{\bar{k}}{M^2}$, $r=p$, $q=2p$, for any $\bar{k}\in[0,k)$;
	
	\item \tnf{(ii)} implies \tnf{(iv)} with $M=\gamma^{p/r}$, $k=\frac{cp}{r}$;
	
	\item \tnf{(iv)} implies \tnf{(ii)} with $r=p$, $\gamma=M$, $c=\frac{\bar{k}}{M}$, for any $\bar{k}\in[0,k)$;
	
	\item \tnf{(iii)} implies \tnf{(iv)} with $p=rs$, $M=(\frac{c_2}{c_1})^s$, $k=s\frac{c_3}{c_{2}^{s}}$;
	
	\item \tnf{(iv)} implies \tnf{(iii)} with $c_{1}=\frac{1}{M}$, $c_{2}=1$, $c_{3}=\bar{k}$, $r=p$, and $s=1$, for any $\bar{k}\in[0,k)$.
\end{enumerate}

\begin{proof}
	We first prove Statements 1), 3), and 5), establishing sufficient conditions for rational stability. 
	
	For Statement 1), suppose $C_i$, $r$, $q$, and
	$V_{\tnf{(i)}}$ satisfy \tnf{(i)}, and let $p$, $M$, and $k$
	be as specified. Define
	\begin{equation*}
		V(x):=
		\left(C_2^{-1}V_{\tnf{(i)}}(x)\right)^{p/r}.
	\end{equation*}
	Since $p=q-r>0$, we have
	\begin{align*}
		V(x)
		&\geq
		\left(\frac{C_1}{C_2}\norm{x}_2^r\right)^{p/r}
		=
		\left(\frac{C_1}{C_2}\right)^{p/r}\norm{x}_2^p
		=M^{-1}\norm{x}_2^p,
	\end{align*}
	as well as
	\begin{equation*}
		V(x)
		\leq
		\left(\frac{C_{2}}{C_{2}}\norm{x}_2^r\right)^{p/r}
		=\norm{x}_2^p.
	\end{equation*}
	Furthermore, using the fact that $\norm{x}_2^q	\geq
	C_2^{-q/r}V_{\tnf{(i)}}(x)^{q/r}$, and $\frac{p}{r}-1+\frac{q}{r}
	=\frac{2p}{r}$, we obtain
	\begin{align*}
		\dot V(x)
		&=
		\frac{p}{r}C_2^{-p/r}
		V_{\tnf{(i)}}(x)^{p/r-1}
		\dot V_{\tnf{(i)}}(x)\\
		&\leq
		-\frac{p}{r}C_3C_2^{-p/r}
		V_{\tnf{(i)}}(x)^{p/r-1}\norm{x}_2^q	\\
		&=
		-\frac{p}{r}\frac{C_3}{C_2}
		\left(C_2^{-1}V_{\tnf{(i)}}(x)\right)^{2p/r}\\
		&=-kV(x)^2.
	\end{align*}
	Thus, $V$ satisfies the conditions of
	Thm.~\ref{thm:LF_sufficient}, and hence $f$ is rationally stable
	on $G$ with power $p$, rate $k$, and gain $M$.
		
	For Statement 3), suppose $c$, $\gamma$, $p$, $r$, and
	$V_{\tnf{(ii)}}$ satisfy \tnf{(ii)}. Define
	\begin{equation*}
		V(x):=V_{\tnf{(ii)}}(x)^{p/r}.
	\end{equation*}
	Then
	\begin{align*}
		V(x)
		&\geq
		\left(\gamma^{-1}\norm{x}_2^r\right)^{p/r}
		=\gamma^{-p/r}\norm{x}_2^p
		=M^{-1}\norm{x}_2^p,
	\end{align*}
	and
	\begin{equation*}
		V(x)
		\leq
		\left(\norm{x}_2^r\right)^{p/r}
		=\norm{x}_2^p.
	\end{equation*}
	Furthermore, since $V(x)\leq\norm{x}_2^p$, we have $-V(x)\norm{x}_2^p\leq -V(x)^2$, and therefore
	\begin{align*}
		\dot{V}(x)
		&=
		\frac{p}{r}
		V_{\tnf{(ii)}}(x)^{p/r-1}
		\dot V_{\tnf{(ii)}}(x)\\
		&\leq
		-\frac{cp}{r}
		V_{\tnf{(ii)}}(x)^{p/r}\norm{x}_2^p
		=
		-\frac{cp}{r}V(x)\norm{x}_2^p
		\leq -kV(x)^2.
	\end{align*}
	Thus, by Thm.~\ref{thm:LF_sufficient}, $f$ is rationally stable
	on $G$ with power $p$, rate $k=cp/r$, and gain
	$M=\gamma^{p/r}$.
		
	For Statement 5), suppose $c_1,c_2,c_3,r,s$ and
	$V_{\tnf{(iii)}}$ satisfy \tnf{(iii)}. Let $p$, $M$, and $k$ be as proposed, and define
	\begin{equation*}
		V(x):=
		\left(c_2^{-1}V_{\tnf{(iii)}}(x)\right)^s.
	\end{equation*}
	Then, the upper and lower bounds on $V_{\tnf{(iii)}}$ in \tnf{(iii)} yield
	\begin{equation*}
		V(x)
		\geq
		\left(\frac{c_1}{c_2}\norm{x}_2^r\right)^s
		=
		\left(\frac{c_1}{c_2}\right)^s\norm{x}_2^{rs}
		=M^{-1}\norm{x}_2^p,
	\end{equation*}
	as well as
	\begin{equation*}
		V(x)
		\leq
		\left(c_2^{-1}c_2\norm{x}_2^r\right)^s
		=\norm{x}_2^p.
	\end{equation*}
	Furthermore,
	\begin{align*}
		\dot{V}(x)
		&=
		sc_2^{-s}
		V_{\tnf{(iii)}}(x)^{s-1}
		\dot V_{\tnf{(iii)}}(x)\\
		&\leq
		-sc_3c_2^{-s}
		V_{\tnf{(iii)}}(x)^{2s}\\
		&=
		-sc_3c_2^s
		\left(c_2^{-1}V_{\tnf{(iii)}}(x)\right)^{2s}
		=
		-sc_3c_2^sV(x)^2
		=-kV(x)^2.
	\end{align*}
	Thus, by Thm.~\ref{thm:LF_sufficient}, $f$ is rationally stable
	on $G$ with power $p=rs$, gain
	$M=(c_2/c_1)^s$, and rate $k=sc_3c_2^s$.
	
	Next, we prove Statements 2), 4), and 6), establishing necessary conditions for rational stability. For each of these statements, suppose that (iv) is satisfied, so that $f$ is rationally stable on $G$ with power $p$, rate $k$, and gain $M$. Then, by Thm.~\ref{thm:LF_necessary}, for any $\bar{k}\in[0,k)$, there exists a continuous $V:G\to\R_+$ such that
	\begin{align*}
		M^{-1}\norm{x}_2^p
		\leq V(x)&\leq\norm{x}_2^p,\\
		\dot V(x)&\leq-\bar{k}V(x)^2.
	\end{align*}
	Now, for Statement 2), let $C_{i}$, $r$, and $q$ be as specified, and let $V_{\tnf{(i)}}:=V$.	Then, the upper and lower bounds on $V_{\tnf{(i)}}$ in \tnf{(i)} follow immediately. Moreover, $V(x)^2	\geq M^{-2}\norm{x}_2^{2p}$, hence
	\begin{equation*}
		\dot V_{\tnf{(i)}}(x)
		\leq-\bar{k}V(x)^2
		\leq-\frac{\bar{k}}{M^2}\norm{x}_2^{2p}
		=-C_3\norm{x}_2^q.
	\end{equation*}
	Thus, \tnf{(i)} holds.
	
	Next, for Statement 4), let $r$, $\gamma$, and $c$ be as defined, and set $V_{\tnf{(ii)}}:=V$. It immediately follows that
	\begin{equation*}
		M^{-1}\norm{x}_{2}^{p}
		=\gamma^{-1}\norm{x}_{2}^{r}
		\leq V(x)
		\leq \norm{x}_{2}^{r}
		= \norm{x}_{2}^{p}.
	\end{equation*}
	Furthermore, we have $V(x)^2\geq M^{-1}V(x)\norm{x}_2^p$, and hence
	\begin{equation*}
		\dot{V}_{\tnf{(ii)}}(x)
		\leq-\bar{k}V(x)^2
		\leq-\frac{\bar{k}}{M}V(x)\norm{x}_2^p
		=-cV_{\tnf{(ii)}}(x)\norm{x}_2^p.
	\end{equation*}
	Thus, \tnf{(ii)} holds.
	
	Finally, for Statement 6), let $c_{i}$, $r$ and $s$ be as proposed, and set $V_{\tnf{(iii)}}:=V$. It follows that
	\begin{equation*}
		c_1\norm{x}_2^r
		\leq V_{\tnf{(iii)}}(x)
		\leq c_2\norm{x}_2^r,
	\end{equation*}
	and
	\begin{equation*}
		\dot V_{\tnf{(iii)}}(x)
		\leq-\bar{k}V_{\tnf{(iii)}}(x)^2
		=-c_3V_{\tnf{(iii)}}(x)^{1+s}.
	\end{equation*}
	Therefore, \tnf{(iii)} holds.
\end{proof}

\end{appendices}

\end{document}